\documentclass[11pt]{article}

\usepackage{amsfonts,amssymb,amsmath,latexsym,ae,aecompl}

\usepackage{xcolor}

\newcommand{\FF}{\vspace*{\medskipamount}}
\newcommand{\FFF}{\vspace*{\bigskipamount}}

\newcommand{\BBB}{\vspace*{-\bigskipamount}}

\newcommand{\cO}{\mathcal{O}}

\newcommand{\Paragraph}[1]{\BBB\paragraph{#1}}
\newcommand{\remove}[1]{}

\newlength{\pagewidth}
\newlength{\captionwidth}
\newcommand{\qed}{\hfill $\square$ \smallbreak}
\newenvironment{proof}{\noindent\textbf{Proof:}}{\qed}

\newcommand{\algoname}[1]{\textnormal{\textsc{#1}}}

\newtheorem{corollary}{Corollary}

\newtheorem{lemma}{Lemma}

\newtheorem{theorem}{Theorem}

\begin{document}

\baselineskip           	3ex
\parskip                	1ex

\title{		Energy Efficient Adversarial Routing in Shared Channels\FFF\FFF\FFF}

\author{	Bogdan S. Chlebus  	\footnotemark[1]   	
		\and
		Elijah Hradovich \footnotemark[2]
		\and
		Tomasz Jurdzi\'nski \footnotemark[3]
		\and
		Marek Klonowski \footnotemark[2]
		\and
		Dariusz R. Kowalski \footnotemark[4]\FFF}

\footnotetext[1]{Department of Computer Science and Engineering, University of Colorado Denver, Denver, Colorado, USA.}

\footnotetext[2]{Katedra Informatyki, Politechnika Wroc\l awska,  Wroc\l aw, Poland.}

\footnotetext[3]{Instytut Informatyki, Uniwersytet Wroc\l awski, Wroc\l aw, Poland.}

\footnotetext[4]{Department of Computer Science, University of Liverpool, Liverpool, United Kingdom.}

\date{}

\maketitle

\begin{abstract}
We investigate routing on networks modeled as multiple access channels, when packets are injected continually.
There is an energy cap understood as a bound on the number of stations that can be switched on simultaneously.
Each packet is injected into some station and needs to be delivered to its destination station via the channel.
A station has to be switched on in order to receive a packet when it is heard on the channel.
Each station manages when it is switched on and off by way of a programmable wakeup mechanism, which is scheduled by a routing algorithm. 
Packet injection is governed by adversarial models  that determine  upper bounds on injection rates and burstiness.
We develop deterministic distributed routing algorithms and assess their performance  in the worst-case sense.
An algorithm knows the number of stations but does not know the adversary.
One of the algorithms maintains bounded queues for the maximum injection rate~$1$ subject only  to the energy cap~$3$.
This energy cap is provably optimal, in that obtaining the same throughput with the energy cap~$2$ is impossible.
We give algorithms subject to the minimum energy cap~$2$ that have latency polynomial in the total number of stations~$n$ for each fixed adversary of injection rate less than~$1$.
An algorithm is $k$-energy-oblivious if at most $k$ stations are switched on in a round and for each station the rounds when it will be switched on are determined in advance.
We give a $k$-energy-oblivious algorithm that has packet delay $\cO(n)$ for adversaries of injection rates less than $\frac{k-1}{n-1}$, and show that there is no $k$-energy-oblivious stable algorithm against adversaries with injection rates greater than~$\frac{k}{n}$.
An algorithm routes directly when each packet makes only one hop from the station into which it is injected straight to its destination.
We give a $k$-energy-oblivious algorithm routing directly that has latency $\cO\bigl(\frac{n^2}{k}\bigr)$ for adversaries of sufficiently small injection rates that are $\cO\bigl(\frac{k^2}{n^2}\bigr)$.  
We develop a $k$-energy-oblivious algorithm routing directly  that is stable   for injection rate $\frac{k(k-1)}{n(n-1)}$, and show that no $k$-energy-oblivious algorithm routing directly can be stable against adversaries with injection rates greater than~$\frac{k(k-1)}{n(n-1)}$.

\vfill

\end{abstract}

\vfill

~

\thispagestyle{empty}
\setcounter{page}{0}

\newpage

\section{Introduction}

\label{sec:introduction}

Energy efficiency has become a critical factor in the design of  large scale distributed and networked systems~\cite{BianzinoCRR12,BollaBDC11,FangYHLL15,JonesSAJ01,OrgerieAL13}.
Ethernet is a popular communication technology used to implement local area networks.
Its energy-efficient standard~\cite{ChristensenRNBMM10} uses adaptive power modes that adjust to the amount of traffic.
Multiple access channels are an abstraction of wireline Ethernet channels.
This work considers communication on such channels that are subject to energy constraints.
The underlying motivation is to design distributed communication algorithms that are efficient with  respect to both packet latency and energy use.

We study dynamic  routing on multiple access channels when packets are injected continually.
A channel is shared by a number of stations.
Without energy considerations, the stations  stay ready all the time to perform their communication tasks.
We add an amount of energy available per round as additional component  of the system.
Formally, there is an upper bound on the number of stations attached to the channel that can be switched on simultaneously, which is interpreted as a cap on the available energy.
Stations that are switched off cannot transmit nor receive messages from the channel, but they can have packets injected into them.

A packet gets injected into some station before it is transmitted on the channel.
A station storing injected packets will strive to transmit them eventually such that they are heard on the channel.
A packet  includes the name of a station it is addressed to. 
If a packet is heard while its destination station is switched off, then handling the packet might be taken over by some other station that would act as a relay for the packet.
A packet can be repeatedly handed over  among the stations such that it hops through a sequence of stations in a store-and-forward manner until eventually it is heard by its destination station, which consumes the packet.

Packet injection is represented by adversarial models that determine upper bounds on injection rates and burstiness~\cite{AndrewsAFLLK-JACM01,BorodinKRSW-JACM01}.
There is no stochastic component in the specification of how packets are injected.
We develop routing algorithms that are distributed and deterministic.
Their performance is measured by the worst-case latency and bounds on the number of queued packets.
Performance bounds depend on the number of stations and the parameters of an adversary controlling packet injections.
Routing algorithms know the size of the system but do not know the parameters of traffic generation.
We consider the classes of routing algorithms defined by additional restrictions.
These restricted algorithms may, for example,  not use relay stations, or not use control bits in messages, or have the on-off status for each station scheduled in advance.

\Paragraph{A summary of the results.}

We develop deterministic distributed algorithms routing adversarial traffic on the multiple access channels and assess their efficiency  in the worst-case sense, where performance bounds depend on the known number of stations $n$ and an unknown adversary.
One of the algorithms maintains bounded queues for the maximum injection rate~$1$ subject only to the energy cap~$3$.
This energy cap is provably optimal, in that obtaining the same throughput with the energy cap~$2$ is impossible.
Algorithms that have bounded latency for each fixed adversary of injection rate less than~$1$ are said to be universal.
We give universal algorithms subject to the minimum energy cap~$2$ that have the latency polynomial in the number of stations~$n$.
One of these algorithms uses control bits in messages and has latency $\cO(n^2)$ and another has stations transmit plain packets only and attains latency $\cO(n^3\log^2 n)$.
An algorithm is $k$-energy-oblivious if at most $k$ stations are switched on in a round and for each station the rounds it is switched on are determined in advance.
We give a $k$-energy-oblivious algorithm that has latency $\cO(n)$ for  adversaries of injection rates less than~$\frac{k-1}{n-1}$ and show that there is no $k$-energy-oblivious stable algorithm against adversaries with injection rate greater than~$\frac{k}{n}$.
An algorithm routes directly when it does not utilize relay stations, in that each packet makes only one hop  straight to its destination from the station it is injected into.
We give a $k$-energy-oblivious algorithm routing directly that has latency $\cO \bigl(\frac{n^2}{k}\bigr)$ for adversaries of sufficiently small injection rates that are $\cO\bigl(\frac{k^2}{n^2}\bigr)$.
We develop a $k$-energy-oblivious algorithm routing directly  that is stable   for injection rate $\frac{k(k-1)}{n(n-1)}$ and show that no $k$-energy-oblivious algorithm routing directly can be stable against adversaries with injection rates greater than~$\frac{k(k-1)}{n(n-1)}$.
All the performance bounds of algorithms and impossibility results are tabulated in Table~\ref{tbl:table}.

\newcommand{\RB}{\raisebox{2.7ex}{}}
\newcommand{\LB}{\raisebox{-1.6ex}{}}

\begin{table}  
\begin{center}
\begin{tabular}{ll|cccc|c}
\hline
\RB\LB{}Algorithm/Impos. &Sec.& Injection & Latency &Queues &Cap & Properties \\
\hline
\hline
\RB\emph{Max throughput:} &&&&&\\
\RB \LB\algoname{Orchestra} &\ref{subsect:algorithm-max-throughput}& $\rho=1$ & $\infty$ & $2n^3+\beta$ &3 &  NObl-Gen-Dir \\
\raisebox{-1ex}{}Impossibility&\ref{subsect:impossibility-throughput}& $\rho=1$ & &Stable &$2$   &      \\
\hline
\RB\emph{Universality:} &&&&&\\
\raisebox{3.2ex}{}\LB\algoname{Count-Hop} &\ref{subsec:general-universal}& $\rho<1$ & $\frac{2(n^2+\beta)}{1-\rho}$ & & 2 & NObl-Gen-Dir\\
\LB\algoname{Adjust-Window} &\ref{subsect:poly}& $\rho<1$ & $\frac{18n^3\log^2n+2\beta}{1-\rho}$ && 2 & NObl-PP-Ind \\
\hline
\RB\emph{Oblivious indirect:} &&&&&\\
\RB\LB\algoname{$k$-Cycle} &\ref{sect:oblivious-indirect}& $\rho< \frac{k-1}{n-1}$ & $(32+\beta)\cdot n$ & & $k$  & Obl-PP-Ind\\
\LB{}Impossibility &\ref{sect:oblivious-indirect} & $\rho> \frac{k}{n}$ &   & Stable &   $k$ & Obl   \\
\hline
\RB\emph{Oblivious direct:} &&&&&\\
\raisebox{3.2ex}{}\raisebox{-2ex}{}\algoname{$k$-Clique} &\ref{sect:oblivious-direct}& $\rho\le \frac{k^2}{2n(2n-k)}$ & $8\frac{n^2}{k}(1+\frac{\beta}{2k})$ && $k$ & Obl-PP-Dir\\
\raisebox{-2ex}{}\algoname{$k$-Subsets}&\ref{sect:oblivious-direct}  & $\rho=\frac{k(k-1)}{n(n-1)}$ & $\infty$ & $2\binom{n}{k} (n^2+\beta)$ &  $k$ & Obl-Gen-Dir  \\
\raisebox{-2ex}{}Impossibility &\ref{sect:oblivious-direct} & $\rho>\frac{k(k-1)}{n(n-1)}$ & & Stable &  $k$ & Obl-Dir   \\
\hline	
\end{tabular}

\parbox{\captionwidth}{\FF\caption{\label{tbl:table} 
A summary of the performance bounds of algorithms and impossibility results, broken into four main sub-topics.
The adversary is of type $(\rho,\beta)$, where $\rho$ is the injection rate and $\beta$ is the burstiness coefficient.
The abbreviations used to specify properties or algorithms or routing are as follows: Obl = oblivious, NObl = non-oblivious, Gen = general, PP = plain-packet, Dir = direct, Ind = indirect.
The impossibilities are of existence of stable routing algorithms subject to the given injection rate, energy cap, and properties of algorithms and routing.
A bound on latency is also a bound on the number of queued packets.
Latency $\infty$ means that it is possible for some packets  never to be delivered.
The parameters $n$ and $k$ can be part of code of algorithms.
}}
\end{center}
\end{table}

\Paragraph{Previous and related work.}

The surveys by Albers~\cite{Albers10} and Irani and Pruhs~\cite{IraniP05} discuss  algorithms for managing energy usage.
Routing and other communication primitives subject to energy constraints have been studied extensively, in particular by Chabarek et al.~\cite{ChabarekSBETW08} and Andrews et al.~\cite{AndrewsAZZ13, AndrewsFZZ10}.
Reducing network energy consumption via sleeping and rate-adaptation was addressed by Nedevschi et al.~\cite{NedevschiPIRW08}.
Bergamo et al.~\cite{BergamoGTMMZ04} proposed distributed power control to improve energy efficiency of routing algorithms in ad hoc networks.

Jurdzi\'nski et al.~\cite{JurdzinskiKZ02} studied the problem of counting the number of active nodes in a single-hop radio network with the goal to simultaneously optimize the running time and the energy spent by each node, which is understood as the length of time interval when a node is awake.
Kardas et al.~\cite{KardasKP13} studied energy-efficient leader election in single-hop radio networks.
Klonowski et al.~\cite{KlonowskiKZ12} considered energy-efficient ways to alert a single hop network of weak devices.
Efficiency of broadcasting in ad-hoc wireless networks subject to the number of transmissions a node can perform, interpreted as energy constraint, was studied in ~\cite{GasieniecKKPS08,KantorP16,KarmakarKPS17}. 
Chang et al.~\cite{ChangDHHLP18} as well as  Klonowski and Paj\k{a}k~\cite{KlonowskiP18}  studied tradeoffs between the time and energy for broadcasting in radio networks.
Chang et al.~\cite{ChangKPWZ17} studied the energy complexity of leader election and approximate counting in models of wireless networks.
Herlich and Karl~\cite{HerlichK} investigated saving power in mobile access networks when  base stations cooperate to be active or passive in extending their range.

Local area data networks implemented by the Ethernet are typically under-utilized; schemes for shutting down network interfaces for energy conservation when using the Ethernet were proposed by Gupta and Singh~\cite{GuptaS07}.
Gunaratne et al.~\cite{GunaratneCNS08} investigated  policies to adaptively vary the link data rate in response to the demand imposed on the link rate as a means of reducing the energy consumption in Ethernet installations.

Adversarial communication in multiple-access channels was studied in~\cite{BenderFHKL-SPAA05,ChlebusKR-TALG12,AnantharamuC15,AnantharamuCR-TCS17}, among others.
Bender et al.~\cite{BenderKPY16} considered the goal of minimizing the number of channel accesses for a constant throughput in the static case of broadcasting when a set of packets is given in advance.
A broadcast algorithm in multiple access channels, that is stable for injection rate~$1$  when the stations are switched  on all the time, was given by Chlebus et al.~\cite{ChlebusKR-DC09}.

\section{Technical Preliminaries}

\label{sec:technical_preliminaries}

We consider multiple access channels as a model of communication networks.
There are~$n$ stations attached to the channel; we refer to the number~$n$ as the \emph{size} of the system.
Each station has a name assigned to it, which is a unique integer in the interval $[0,n-1]$.
The network operates in a synchronous manner, in that an execution of a communication algorithm is partitioned into rounds.
All the stations begin an execution of an algorithm in the same round.

\Paragraph{Messages.}

Stations may transmit messages on the channel.
The duration of a round and the size of a message are mutually scaled  such that it takes a round to transmit one message.
We say that a station \emph{hears} a transmitted message when the station is switched on and it receives the transmitted message successfully as feedback from the channel.
If exactly one station transmits a message in a round then all the stations that are switched on in this round hear the message, including the transmitting station.
When at least two stations transmit their messages in the same round then no station can hear any message in this round, including the transmitting stations.
A round during which no message is transmitted is called \emph{silent}.

A station can be in one of two possible modes in a round.
When a station is \emph{switched on}  then this means that it is fully operational, in being able to transmit a message and receive feedback from the channel.
When a station is \emph{switched off} then this means that it cannot transmit nor receive any communication from the channel.
Each station is autonomous in when it is on and when it is off.
In a round in which a station is switched on, the station can set its \emph{timer} to any positive integer $c$, which results in the station spending the next $c$ rounds in the off-mode and returning to the on-mode immediately afterwards.

It is a natural goal to design ways to save on the expenditure of energy while the  communication system would maintain its functionality.
We associate energy cost as an attribute of the communication infrastructure in the following manner.
All the stations are connected to an external energy source, which has some output capacity and cannot provide more energy per round than this cap.
Specifically, the following is assumed: (1)~it costs one energy unit to keep a station switched on in a round, and (2)~it costs a negligible amount of energy to keep a station switched off in a round.
When representing the whole system's expenditure of energy in a given round, we make it equal to the number of stations that spend this round switched on.
The upper bound on the number of stations that can be  switched on simultaneously in a round  is the \emph{energy cap} of the system.
A multiple-access channel system is determined by the total number of available stations and the energy cap.
The energy cap~$2$ is minimum to make the tasks of point-to-point communication feasible in principle, since at least one transmitter and one receiver need to be switched on in a round.

\Paragraph{Dynamic packet generation.}

A packet~$p=(d,c)$ consists of its \emph{destination address~$d$} and its \emph{content~$c$}.
A destination is an integer in $[0,n-1]$ interpreted as a name of the station to which the packet needs to be delivered.
A packet's contents is the information that the packet carries, which does not effect how the packet is handled.

An adversarial model imposes quantitative constraints on how packets get generated and  injected.
An adversary is determined by its \emph{type $(\rho,\beta)$}, where $\rho$ and $\beta$ are numbers such that $0<\rho\le 1$ and $\beta\ge 1$.
In each continuous time interval of length~$t$, the adversary can generate and inject at most $\rho\cdot t + \beta$ packets.
The parameter $\rho$ is interpreted as an \emph{injection rate}.
The maximum number of packets that can be generated in a
single round is the adversary's \emph{burstiness}, which is~$\lfloor \beta+\rho \rfloor$; we call $\beta$ the \emph{burstiness coefficient}.
When referring to an upper bound $\rho\cdot t + \beta$ on the number of injected packets, we say that $\rho t$ is the \emph{injection-rate component} and $\beta$ is the \emph{burstiness component} of the bound.
This adversarial model of packet injection is called \emph{leaky bucket}; it was used before to model traffic in shared channels, in particular in~\cite{ChlebusKR-DC09,AnantharamuCKR-JCSS19}.

An adversary is restricted only by the number of packets it can generate in a time interval, as determined by its length.
Once a packet is generated, the adversary injects it immediately into some station. 
Packets may be injected into any station, regardless of whether the station is switched on or off.
Packets injected into a station can be stored in the station's private memory, which is called this station's \emph{queue}. 
A station can transmit its queued packets in arbitrary order.
A station can scan its queue and access any packet in negligible time.

\Paragraph{Routing packets.}

Each injected packet needs to be delivered to its destination station.
We say that a packet $p=(d,c)$ gets \emph{delivered} in a round~$t$ when the following occurs: (1) a message containing packet~$p$ is transmitted  in round~$t$ and is heard on the channel, (2) the destination station~$d$ is switched on in round~$t$.
If a packet gets delivered then it is ``consumed'' by the destination station and disappears from the system.
A packet may be transmitted and heard on the channel a number of times, which results in the packet hopping from station to station in a store-and-forward manner.
If a station transmits a packet, which is then heard on the channel, then the packet may be removed from the queue of the transmitting station.
If a message with a packet is heard on the channel but the packet is not delivered in this round then some station may \emph{adopt} the packet by adding it to the queue; such a new station handling the packet becomes a \emph{relay} for the packet and treats it as if it were injected directly by the adversary.

The task of \emph{routing} is defined as follows: while packets are continually generated and injected into the system, stations transmit them such  that they are eventually delivered. 
The total number of packets that are queued in a round is referred to as the \emph{queue size} in this round.
The \emph{delay of a packet~$p$} is defined as the difference $t_2-t_1$ between the round~$t_2$ in which packet~$p$ gets delivered and  the round~$t_1$ in which packet~$p$ was injected.

\Paragraph{Routing algorithms.}

Routing is performed by distributed algorithms that are executed by all the stations concurrently.
Correctness of a routing algorithm means that each injected packet is eventually delivered to its destination and a delivery occurs exactly once for each packet.

A station switched-on  in a round either transmits a message or senses the channel (listens to it) in this  round.
A message  consists of at most one packet and  a string of control bits.
The bits encoding packet's destination address are not considered as control bits.
Whether control bits are included in messages is a feature of algorithms.
Algorithms that have a message consist of only a packet without any control bits are called \emph{plain-packet} ones, they make a subclass of \emph{general} routing algorithms.
A station executing a plain-packet algorithm cannot send a message when it does not have a packet to route, since a message has to include a packet and nothing else. 
A station executing a general-routing algorithm may transmit a message without any packet but with control bits only.

The destination address of a packet is just a station's name, so it is represented by $\cO(\log n)$ bits.
We consider only routing algorithms that use the conservative amount of $\cO(\log n)$ control bits per message.
This restriction on the number of control bits transmitted per round makes  coordinating actions among the stations reasonably costly, as reflected in the number of messages, and ultimately packet delays and size of queues.

Routing algorithms that do not use relay stations are said to \emph{route directly}, and otherwise they are said to \emph{route indirectly}.
Algorithms that route directly make each packet hop only once, from the station into which the packet got injected, straight to the destination.

A routing algorithm is called \emph{energy oblivious} when it determines in advance, prior to the start of an execution, for each station~$i$ and each round $t$, whether station~$i$ is on or off in  round~$t$.
An energy-oblivious algorithm executed on a channel subject to an energy cap of at most~$k$ is called \emph{$k$-energy-oblivious}.

The \emph{queue size} measure, of an execution of a routing algorithm, is defined as a maximum number of queued packets in a round of this execution.
The \emph{latency} measure, of an execution of a routing algorithm, is defined as the maximum packet delay occurring in the execution.
Both the queue size and latency are natural performance metrics of routing algorithms and are represented as functions of the size of the system and the type of an adversary.
If the latency of a routing algorithm is bounded then queues are bounded as well, since a queue's size at a station is always a lower bound on the delay of some packet handled by this station.
We say that a routing algorithm is \emph{stable}, against a class of adversaries, if the queue size is bounded, for a given number of stations and an adversary in this class.
The \emph{throughput} of an algorithm is the maximum injection rate for which it is stable, assuming such a rate exists.
Throughput is always at most~$1$ since at most one packet can be heard in a  round.

A routing algorithm has a \emph{universally bounded latency} when latency is bounded against each adversary of injection rate less than~$1$; we call such algorithms \emph{universal}.
Bounds on the latency of a routing algorithm against a specific adversary involve the size of the system $n$ and the type of an adversary $(\rho,\beta)$.

\Paragraph{Knowledge.}

We say that a property of a system is \emph{known} when it can be used in codes of algorithms.
It is assumed throughout that the size of the system~$n$ and the energy cap $k < n$ are known, but the adversary's type $(\rho,\beta)$ is not.
Algorithms may have their correctness and performance bounds depend on the magnitudes of the unknown parameters of the communication environment.  
For example, an algorithm may be stable or have bounded latency for sufficiently small injection rates.

\section{Maximizing Throughput}

\label{sec:maximizing-throughput}

We present a direct-routing algorithm stable for injection rate~$1$.
This is the maximum throughput possible on multiple access channels.
The algorithm requires energy cap to be at least~$3$.
We show  that the number $3$ is best achievable in this sense by proving the impossibility of attaining throughput~$1$ with energy cap~$2$.

\subsection{An algorithm achieving maximum throughput}

\label{subsect:algorithm-max-throughput}

An algorithm stable for injection rate~$1$ which we give is called \algoname{Orchestra}.
It schedules at most three stations to be simultaneously switched on  at any round, with at most one of them transmitting.
The algorithm builds on the paradigms developed in~\cite{ChlebusKR-DC09}, which gave a broadcast algorithm with throughput~$1$ for multiple-access channels without energy caps.

We call a group of $n-1$ consecutive rounds of an execution a \emph{season} if the last round $t$ of the group satisfies $t \equiv 0 \pmod{n-1}$.
For each season, there is a unique station associated with it called a \emph{conductor}.
Stations that are different from a conductor  are called \emph{musicians}.
A conductor for a season transmits a message in every round  of  this season, so there are no silent rounds.
A round when a message heard on the channel does not include a packet but only control bits is called \emph{light}.

Every station keeps an ordered list of all the stations. 
These lists are the same in every station at the beginning of a season; at such a moment they represent one list, which we call the \emph{baton list}.
Initially, the baton list consists of all the stations ordered by their names.
Stations assume the role of conductors in their order on the baton list.
The first station on the list is assigned to serve as a conductor for the first season.

The positions of stations on the baton list are understood as follows: the front entry of the list is considered as the first station on the list, then the consecutive stations have their positions increased by one, and the tail entry occupies the last $n$th position.
In particular, if a station at position~$i$ moves to become the head of the baton list, then this stations acquires position~$1$ while each station at the original position $j<i$, which means closer to the front than~$i$, gets its position incremented to~$j+1$, so that its distance from the head of the list increases.

The process of assigning conductors to seasons can be visualized as passing a virtual baton from station to station, such that a station holding the baton  is  a conductor.
When a season ends then the baton is typically passed on to the next station on the baton list.
The order determined by the list is understood in a cyclic sense, in that the first station becomes a conductor after the last one in the list has concluded its assignment.
An exception for this process occurs when a conductor is moved to become the head of the baton list while keeping the baton.

A conductor of a season is switched on in each round of the season.
A musician switches on during a season either to learn or to receive, or possibly both.
A message transmitted by a conductor contains control bits for a taught musician who is to learn, and a packet sent to a receiving musician, unless the message is light.
We explain the actions of teaching/learning and sending/receiving next.

The purpose of \emph{learning} is to obtain information from the conductor, in particular one pertaining to a schedule to receive packets in  the next season with the same conductor.
A station learns by extracting control bits from a message transmitted by the conductor and interpreting them as round numbers to be switched on during the next season when the same stations acts as conductor.
If a conductor conveys information to a learning station then we say that the conductor \emph{teaches} the station this information.

The purpose of \emph{receiving} is to obtain packets injected into a conductor and destined for a receiving musician.
A musician receives by extracting a packet included in a message from the conductor.
If a conductor transmits a message with a packet then it \emph{sends} this  packet to the receiving station.

For a message transmitted by a conductor to effectively serve its purpose, the following three involved stations need to be switched on simultaneously: the conductor, the learning musician, and the receiving musician.
Musicians  switch  on during a season in the following manner.
First, the musicians switch on to learn: they do it one by one, in the order of their names, for one round at a time.
Second, the musicians switch on to receive: they do it according to the schedule taught  during the latest previous season when the same station was a conductor.

A packet injected into a station, when it acts as a conductor, stays \emph{new} for the duration of this season, and after that becomes \emph{old}.
A packet injected into a musician becomes old immediately.
In particular, when a new season begins, then all the packets queued in the stations are old.
At the start of a season, when some station becomes a conductor, this station  computes a schedule to send the old packets during the next season when it will become a conductor again.
The schedule concerns only these old packets that have not been scheduled yet for the current season.
A conductor schedules packets to send  in the order of their injections.

A station considers itself \emph{big} if it has at least $n^2 -1$ old packets in its queue.
A conductor that is big at the beginning of a season teaches each musician of this status, by suitably setting a toggle bit in messages.
After a musician learns this information, it moves the conductor to the front of its private copy of the baton list.
Such a season concludes with all the musicians having identical private lists representing the baton list, with the conductor at the front.
A big conductor  keeps the baton for the  next season, after moving to the front of the baton list, and stays at the front as long as it is big.
This mechanism allows for one station to act as a conductor for long stretches of seasons, possibly indefinitely, should the adversary inject packets into one station only.

We group seasons into contiguous intervals of seasons, depending on the  heaviness on traffic  during these seasons.
If the total number of packets in the queues at the beginning of a season is greater than $ n^3-2n+1$ then the season belongs to a \emph{dense  interval of seasons}, which means the traffic will be heavy.
Otherwise, a season belongs to a \emph{sparse interval of seasons}, which means that traffic might be light.
We expect that there exist big stations when traffic is heavy.
A station is \emph{pre-big} in a round of an interval of seasons if it has not been big during this interval before the round.
A station is \emph{post-big} in a round of an interval of seasons if it is not big now but it has been big by this round during the interval.

%: theorem

\begin{theorem}

There are at most $2n^3+\beta$ packets queued in a round of an execution of algorithm \algoname{Orchestra} against an adversary of injection rate~$1$ and with a burstiness coefficient~$\beta$.
\end{theorem}

\begin{proof}
Let $D = n^3-2n+1=n(n^2-2)+1$ be the number of old queued packets used to differentiate between dense and sparse intervals of seasons.
A big station has at least $n^2-1$ old packets in its queue.
By the pigeonhole principle, there exists at least one big station during a season in a dense interval of seasons.

We estimate the number of queued packets during a season depending on whether the season belongs to a dense or sparse interval.
The adversary's capability to inject packets due to the burstiness coefficient is accounted for only once at the end of a derivation of an upper bound on the number of queued packets. 

First, we consider sparse intervals.
The system starts a season with empty queues, so the first season belongs to a sparse interval.
If the system starts a season in a sparse interval then it has at most~$D$ packets.
The adversary can inject at most $(n-1)$ packets during a season.
So there can be at most these many old packets in stations within a sparse interval:
\begin{equation}
\label{eqn:sparse}
D + (n-1)  = n^3-2n+1 + n-1 = n^3 -n
\ ,
\end{equation}
not including burstiness.
This is also an upper bound on the number fo queued packets when a 
 sparse interval ends and a dense interval begins.

Next, we consider dense intervals.
In such an interval, the adversary can be assumed to inject at full power, namely,  a packet per round.
If a message with a packet is heard on the channel, then this does not affect the number of queued packets, since only one packet gets injected.
Otherwise, if a light message is heard, this results in the number of queued packets incremented by~$1$.
This makes an upper bound on the number of light messages heard on the channel serve as an upper bound on the increase on the number of queued packets.

We claim that neither big nor post-big stations can contribute light rounds when acting as conductors during dense intervals.
It follows that only pre-big stations contribute light rounds.
We prove this claim next. 

A big station has at least $n^2-1$ packets in its queue at the beginning of a season when it obtains the baton.
It must have had at least $n-1$ old packets to schedule at the beginning of the previous season it was conducting, since the adversary could inject at most these many packets in the meantime during $n$ seasons, without accounting for burstiness:
\[
n(n-1)=n^2 -1-(n-1)
\ .
\]
A conductor that had at least $n-1$ old packets at the beginning of the previous season it conducted, has already scheduled a full season,  so it sends a packet in each among the $n-1$ rounds of the current season.
We conclude that a big station contributes no light  rounds as long as it gains and maintains the status of a big conductor.

Consider a station~$i$ that is post-big but not big.
Station~$i$ has an opportunity to transmit only when there is no big station before it on the list, as such a station would be visited by the baton first and moved the baton back to front.
When station $i$ receives the baton and $i$ is not big, then there is a big station after~$i$, because such a station exists in every season of a dense interval.
The first big station encountered by the baton is moved to the front of the baton list, thereby incrementing the $i$'s position to~$i+1$.
The position of a station cannot increase more than $n-1$ times in this way.
The last time when $i$ was big, it had at least $n^2 -1$ packets in its queues and was placed at the front of the baton list.
Station~$i$ had at least these many packets  when ending a season in which it was a big conductor:
\[
n^2-1-(n-1)= n(n-1)
\ .
\] 
So the station can afford to increase its position up to $n$ times while consistently sending $n-1$ packets per season when serving as a conductor. 
We conclude that a post-big station contributes no light  rounds during seasons in a dense interval.

We are finally ready to count light rounds in a dense interval, all of which could be contributed by pre-big conductors only.
There are at most $n-1$ pre-big stations in the system at the beginning of a dense interval, because at least one station is big.
Light rounds occur in a season when the conductor has fewer than $n-1$ old packets scheduled to send; let us assume conservatively that such a conductor does not have any scheduled packets to send, to maximize the number of light rounds the season contributes.
A pre-big station $i$ becomes a conductor only when the first big station  on the baton list is behind station~$i$.
After the baton leaves and eventually reaches a big station, this big station advances to the front and the $i$'s position shifts by~$1$.
Such shifts can occur at most $n-1$ times.
It follows that a pre-big station can contribute at most $(n-1)^2$ light rounds when acting as a conductor in a dense interval.
Therefore all the pre-big stations together contribute at most $(n-1)^3$ light rounds as conductors.

The bound~\eqref{eqn:sparse} estimates the number of queued packets when a dense interval begins.
This number can grow by at most the number of light rounds while the adversary injects at full power, plus burstiness.
These three parts together contribute the following:
\[
n^3-1 +(n-1)^3 +\beta = (n-1)(2n^2  -2n +1)+\beta = 2n^3 -4n^2 +3n -1+\beta\le 2n^3+\beta
\ .
\]
This quantity serves as the ultimate upper bound on the number of queued packets.
\end{proof}

\subsection{An impossibility for maximum throughput}

\label{subsect:impossibility-throughput}

Algorithm \algoname{Orchestra} requires at least $3$ stations to be switched on in each round.
We show that this is necessary for any algorithm to have throughput~$1$.

%: lemma

\begin{lemma}
\label{lem:executions}

Given an algorithm for a system of  $n \geq 3$ stations, let us assume that we have defined an execution of the algorithm until some round $t_{i-1}$ such that the following holds: at least one station~$s$ has no packets in its queues, no other station has packets to be delivered to~$s$, and there are at least~$i-1$ packets in the system.
Then either the execution can be extended without bounds and the number of packets in the system grows unbounded or there exists a round $t_i>t_{i-1}$ such that the execution can be extended until $t_i$ in a way that the round $t_i$ satisfies the following conditions:
\begin{enumerate}
\item[\rm (a)] no packet is successfully transmitted at round $t_i$, and

\item[\rm (b)] by the end of round $t_i$ at least one station~$s'$ has no packet in its queues and no other station has packets destined for station~$s'$, and

\item[\rm (c)] after round $t_{i-1}$ and by the end of round~$t_i$, one packet per round has been injected into the system on average and the burstiness was at most $1$.
\end{enumerate}
\end{lemma}

\begin{proof}
Let $t_0$ be the first round, and $i$ be an integer such that $i>0$.
Let $s_1$ and $s_2$ be two stations different from $s$.
We consider the following two possibilities to extend an execution, as determined by the adversary after round $t_{i-1}$.
\begin{description}
\item[\rm Case~I:] No packet is injected into station~$s$, station~$s_1$ gets  one packet injected  into it addressed to~$s$ in each odd round and one packet addressed to~$s_2$ in each even round:

\item[\rm Case~II:] No packet is injected into station~$s$, station~$s_1$ gets one packet injected  into it addressed to station~$s_2$ in each round.
\end{description}

For station~$s$, these two cases to extend the execution are indistinguishable up to a round~$t$ when station~$s$ becomes switched on for the first time after round~$t_{i-1}$.
Now there are two possible continuations:
\begin{description}
\item[\rm Continuation~1:] such a round~$t$ does not exist; in the  execution determined by Case~I the number of packets addressed to station~$s$ grows unbounded.

\item[\rm Continuation~2:] such a $t$ exists; then the execution determined by Case~II  extended to round~$t_i = t$ satisfies the following:
\begin{enumerate}
\item[\rm (a)] 
no packet is heard at round $t$, as there are no packets involving station~$s$ in the system;

\item[\rm (b)] 
station~$s$ has no packets in its queues, since it had no packets in round~$t_{i-1}$ and no packet in the system was addressed to it, and between the rounds~$t_{i-1}$ and~$t_i$ (inclusive) this has not changed;

\item[\rm (c)] 
the adversary  injects exactly one packet per round.
\end{enumerate}
\end{description}
By the properties (a) through~(c) above, and by the assumption that there are at least $i-1$ packets in the system by round~$t_{i-1}$, the number of packets  at the end of round~$t_i$ is at least~$i$.
\end{proof}

Lemma~\ref{lem:executions} gives a sequence of rounds $(t_i)_{i\ge 0}$ such that there are at least~$i$ packets queued in the system at round~$t_i$.

%: theorem

\begin{theorem}
No  algorithm can be stable for energy cap $2$ and a system size greater than or equal to~$3$ against leaky-bucket adversaries with injection rate~$1$.
\end{theorem}

\begin{proof}
Suppose that such an algorithm exists, to arrive at a contradiction.
The argument is by induction on the round numbers.
Consider the first round of an execution of the algorithm, to provide a base for induction.
By the assumption about the energy cap, at least one station~$s$ needs to be switched off.
The assumptions of Lemma \ref{lem:executions} are satisfied for~$t_0=1$, so that $i=1$.

Next we show the inductive step.
Assume that the assumptions of Lemma \ref{lem:executions} are satisfied for some~$i\ge 1$, so that  there is an execution  determined up to some round $t_{i-1}$ such that the following holds: at least one station~$s$ has no packets, no station has packets addressed for~$s$, and there are at least $i-1$ packets in the system.

By Lemma \ref{lem:executions}, this prefix of an execution up to round~$t_{i-1}$ either could be extended to a full execution such that  the number of queued packets grows unbounded, or
there is a round~$t_i$ and an extension that satisfy the assumption of Lemma \ref{lem:executions}
for~$i+1$.
To conclude, either there is $i$ such that from round~$t_i$ there is an~unstable extension of the execution, or we could continue extending the execution through rounds~$t_j$, for all integers~$j$.
In the latter case, the number of packets in the system grows unbounded with~$j$, and the resulting execution is unstable.
\end{proof}

\section{Universal Routing}

\label{sec:universal}

We give two routing algorithms with universally bounded latency.
One routes directly while using control bits in messages to coordinate stations and the other routes indirectly with messages consisting of plain packets only.

\subsection{General universal routing}

\label{subsec:general-universal}

The direct-routing algorithm using control bits in messages is called \algoname{Count-Hop}; it operates as follows.
One station is dedicated to serve as a \emph{coordinator} and the other stations are \emph{workers}.
An execution is structured into \emph{phases}.
Packets transmitted in a phase need to be \emph{old}, in that they were injected in the previous phase.
Packets injected in the current phase are \emph{new} for the duration of the phase.
At a round when a phase ends, all the packets available in the system become old for the next phase.
These are the only old packets for the next phase, which means that each station knows which among its packets are old, for the duration of this phase, when a new phase begins.
The first phase consists of $n$ rounds during which all the stations are switched off.
Each of the following phases proceeds through \emph{stages}, which are time intervals spent by the stations working to accomplish some task.

A phase is partitioned into $n$ stages, one for each receiving station.
Such a stage for each receiving station consists of three substages.
During the first substage, each station, except for the receiving station~$v$ and the coordinator, transmits once, sending a message with the number of old packets destined for~$v$.
This information allows the coordinator to assign to each station a time interval to transmit all its packets to~$v$.
The second substage consists of the coordinator transmitting the offset number for each station to be switched off waiting for its turn to transmit.
Finally, the third substage has all the stations switch on  one by one when the turn comes to transmit the old  packets destined to~$v$, while the station~$v$ is switched on during the whole substage and the coordinator is switched off.

%: theorem

\begin{theorem}
\label{thm:one-hop}

A direct-routing algorithm \algoname{Count-Hop} requires the energy cap $2$, is stable for each injection rate $\rho<1$, and its latency for such an injection rate is at most the following:
\[
\frac{2(n^2+\beta)}{1-\rho}
\ .
\]
\end{theorem}

\begin{proof}
Each packet is delivered from the station of injection to the station of destination in one direct hop, by the algorithm's design.
There are $(n-1)^2$ rounds spent transmitting numbers during a phase but no packets.
While such messages are transmitted, the adversary can inject new packets.
These packets will extend the duration of the next phase by up to $\rho (n-1)^2$ rounds.
This phenomenon can be iterated in a cascade-like manner, since when packets are transmitted, the adversary can use this time to inject even more packets.
Taking into account all the possible extensions of phases, the duration of any phase is at most the following:
\[
(n^2+\beta)(1+\rho + \rho^2+\ldots ) = \frac{n^2+\beta}{1-\rho}
\ .
\]
A packet stays in a station during at most two consecutive phases.
\end{proof}

\subsection{Plain-packet universal routing}

\label{subsect:poly}

We describe an indirect-routing plain-packet algorithm that requires only a constant energy cap, but has universally bounded latency and attains packet delay $\cO(n^3\log n)$.
The algorithm is called \algoname{Adjust-Window}.

An execution of  \algoname{Adjust-Window} is structured into segments called \emph{time windows}.
The size of a time window may increase in the course of an execution.
The current size of a window is denoted by~$L$. 
All the stations use the same value of $L$ at each round. 

Packets injected before the current time window are called \emph{old} and  packets injected during the current time window are called \emph{new} for the duration of the window.
The goal to achieve during a window is to deliver all the outstanding old packets to their destinations. 
Whether or not this goal is accomplished in a particular time window may depend on the magnitude of~$L$.
Old packets that do not get delivered in a window remain old for the duration of the next window.  
If some old packets are not delivered in the current time window, then the window size~$L$ gets doubled to become~$2L$, which determines the duration of the next window. 
Otherwise, if all the old packets are successfully delivered in a window, then the window size~$L$ stays the same,  and so the duration of the next window stays the same as well.

Algorithm \algoname{Adjust-Window} works with energy cap~$2$.
An execution is organized such that in each round at most one station transmits.
If a station~$i$ transmits a message with packet in a round and another station~$j$ is switched on, then we say that \emph{station~$i$ sends the transmitted packet to~$j$}.
If a station~$i$ sends a packet to station~$j$ and the packet is heard on the channel, then station~$i$ removes the packet from its queue and station~$j$ either consumes it, if it is addressed to~$j$, or else adopts it and becomes its relay station.
This means  that packets may hop from station to station, and routing may be indirect.

A time window is partitioned into three stages: Gossip, Main, and Auxiliary.
The goal of a Gossip stage is to exchange information between stations regarding the numbers of old packets in their queues with particular destinations. 
In a Main stage, the stations transmit old packets directly to their destinations according to a schedule based on the information exchanged and learned during the preceding Gossip stage. 
A station knows the part of such a schedule relevant to its actions:  it knows   when to transmit messages to which destinations, and in which rounds to listen to messages addressed to it. 
It may happen that a station~$i$ needs to convey some information to a station~$j$ while station~$i$ does not have packets with the destination~$j$, then $i$ sends some packet(s) to~$j$ whose destination is different from~$j$. 
An Auxiliary  stage deals with delivering such relayed packets to their destinations, as well as handling old packets at stations that could not participate in neither Gossip nor Main stages due to lacking sufficiently many packets.

A message transmitted on the channel may only contain one plain packet  without attached control bits.
This means that numbers encoded as strings of bits cannot be piggybacked on messages.
Instead,  we design a protocol called \emph{coded transfer (of bits)} to encode sequences of bits by way of sequences of transmissions of single packets, with rounds of transmissions possibly interspersed with silent rounds. 
One round of coded transfer can convey one bit.
Coded transfer needs to overcome the following technical obstacle: a station that is supposed to transmit a packet to convey a bit needs to have at least one packet available in its queue, which after a successful transmission is removed from the queue. 
This implies that stations with empty queues cannot transmit messages and so their lack of transmission activity needs to be properly interpreted by the other stations.

Coded transfer of bits works as follows.
Suppose that a station~$i$ is to transfer $r$ bits $B_1,B_2,\ldots,B_r$ to another station~$j$, for $0\le i,j< n$, and the size of the queue of station~$i$ is at least~$r$. 
Then, in $r$ consecutive rounds, station~$i$ sends a packet to~$j$ in the $k$th consecutive round if and only if $B_k=1$, for $1\le k\le r$, while station~$j$ listens to the channel.
This approach makes the transmitting station~$i$ use  one packet for each transmitted bit~$1$ and no packet for a~$0$. 
Station~$i$ may transmit packets not addressed to~$j$, if packets addressed to~$j$ are not available at~$i$; if this occurs then station~$j$ adopts them and becomes their relay.

The stages of a window take a specific duration, depending on $n$ and $L$.
Let $L_G$, $L_M$, $L_A$ denote the number of rounds of a Gossip stage, a Main stage and an Auxiliary stage, respectively.
These three numbers sum up to $L$: $L_M+L_G+L_A=L$.
The magnitudes of $L_G$ and $L_A$ are determined next, and the remaining part of $L$ rounds of a window is taken by a Main stage.

We specify that stations without sufficiently many old packets do not participate either in Gossip stages or in Main ones. 
We categorize such stations as small.
In what follows, the notation~$\lg x$ stands for~$\lceil \log_2  (x+1)\rceil$.
Formally, a station is \emph{small} in the considered window of size $L$ if the size of its queue at the beginning of that time window is less than $4n\lg L$; otherwise, the station is \emph{large} in the window.
A large station has sufficiently many packets to transmit should they be needed to convey bits by coded transfer.

\Paragraph{A Gossip stage.}

The goal of a Gossip stage is to share information among the stations about the contents of their queues at the beginning of the current time window. 
Such transmission of information is performed  indirectly by coded transfer.

A Gossip stage consists of $n^2$ \emph{phases}, indexed by all pairs $(i,j)$ for $1\le i,j\le n$.
Each phase takes $2+3\lg L$ consecutive rounds.
Thus, a Gossip stage takes these many rounds:
\[
L_G=n^2(2+3\lg L)
 \ .
\]
An $(i,j)$-phase for $i\neq j$ is structured as follows. 
The station~$j$ listens to the channel in each round of a phase, as the only station that does so. 
If the station~$i$ is small, it stays silent for the whole phase; otherwise, if $i$ is large,
it conveys some information to $j$ as follows. 
The station~$i$ sends a packet to $j$ in the first round of the phase to notify $j$ that $i$ is large. 
Then, in the second round of the phase, $i$ sends a packet to $j$ if and only if its queue size is greater than~$L$.
Finally, during the following $3 \lg L $ rounds of the phase, $i$ conveys the following three numbers to $j$ by coded transfer: 
\begin{enumerate}
	\item
	the minimum of its queue size and $L$, 
	\item
	the number of packets in its queue with destination~$j$, or $L$ if the number of packets to $j$ is at least $L$, 
	\item
	the number of packets in its queue with destinations $k$  such that $k<j$, or $L$ if the number of such packets is at least~$L$.
\end{enumerate}
At the end of a Gossip stage, each station~$j$ knows one of the following about each station~$i$, where $0\le i< n$ and the size of the queue of station~$i$ is measured at the beginning of a Gossip stage:

\begin{enumerate}
	\item[(a)] the queue size of $i$ is less than $4n\lg L$, or otherwise 
	\item[(b)] the queue of $i$  has more than $L$ packets to some destination, or otherwise
	\item[(c)] the exact size of the queue of $i$, the number of packets in $i$ with destinations $k$ such that $k<j$, and the number of packets in the station~$i$ with the destination $j$, when none of the cases~(a) nor~(b) holds.
\end{enumerate}

This information determines the size of the next time window. 
Namely, if there is some station~$i$ such that the queue size of~$i$ is greater than~$L$ then the window size is doubled to become~$2L$. 
Similarly, if none among the queue sizes is greater than $L$ but the sum of the queue sizes of all the stations is greater than the length of the Main stage, the window size is also doubled to become~$2L$. 
If none of the two conditions holds, the time-window size $L$ stays the same for the duration of the next time window.

\Paragraph{A Main stage.}

If it is known that some stations have their queue sizes greater than~$L$, then the Main stage is dedicated to the station with the smallest name among them, which spends all the rounds transmitting its packets.
Suppose otherwise that no station has the size of the queue greater than~$L$.
Let $m$ be the total number of packets queued in the stations, which is known by each station.
The stations that are small in this window, meaning with fewer than~$2+3n \lg L $ packets in queues at the beginning of the window, do not transmit in this stage, as if they had no packets.
Based on the information collected in the Gossip stage, every station can compute on its own a comprehensive schedule for delivering the minimum of $L_M$ and $m$ packets from their queues that have been already stored in these queues at the beginning of the current time window. 
The schedule determines the sender of a packet and the destination of a packet for each round. 
Transmitting according to such a schedule completes the stage, where only a transmitter and receiver are switched on in each round.

A Main stage, as given above, has stations operate based on the sizes of their queues at the beginning of the current time window.
The actual numbers of old packets that the stations have in their queues, when a Main stage was planned, might have changed during the Gossip stage. 
This is because, as stations transmit old packets during a Gossip stage, these packets are not necessarily received by their destination stations, and so still need to be forwarded by the stations that received them and now should act as relays. 
This issue is taken care of by Auxiliary stages.

\Paragraph{An Auxiliary stage.}

The goal of this stage is to deliver all the old packets  that  are in the queues of small stations along with the packets received by the stations in a Gossip stage during the coded transfer that still need to be forwarded.
This task is accomplished by the following round-robin style algorithm.

A stage is structured into \emph{phases} of $n^2$ rounds each, indexed by the pairs $(i,j)$ for $0\le i,j <n$.
In a round $(i,j)$ of a phase, $j$ listens and $i$ sends a packet to~$j$, provided that $i$ has such a packet in its queue. 
Since each small station has at most  $4n \lg L$ packets at the beginning  of a window,  and a station can receive at most these many packets 
\[
(2+3\lg L)\cdot(n-1)\le 4 n\lg L
\]
during a Gossip stage, provided that $2\le \lg L$, it is sufficient to execute $8n\lg L$ phases to guarantee that all the considered packets are delivered to their destinations.
We may specify that an Auxiliary stage takes these many rounds:
\[
L_A=n^2\cdot 8n\lg L
\ .
\]

To summarize, a Gossip stage consists of $n^2(2+3\lg L)\le 4n^2\log L$ rounds and an Auxiliary stage takes  $8n^3\log L$ rounds.
A Main stage takes the remaining rounds, their number being at least
\[
L-4n^2\lg L - 8n^3\lg L
\ge 
L-9n^3\lg L
\ ,
\]
for sufficiently large~$n$.
We set the initial value of $L$ to the smallest natural number such that the following inequality holds:
\[
L-9n^3\lg L\ge\frac12L
\ .
\]
Thus the first Main stage takes at least half of the length of the first window, and so there is enough room for the first Gossip and Auxiliary stages to be completed.

%: theorem adjust window

\begin{theorem}
\label{thm:adjust-window}

A plain-packet algorithm \algoname{Adjust-Window} needs the energy cap~$2$ and has the following latency, for each adversary of injection rate $\rho<1$ and burstiness $\beta$:
\[
\frac{18n^3\log^2n+2\beta}{1-\rho}
\ , 
\]
where $n$ is sufficiently large with respect to $\rho$ and $\beta$.
\end{theorem}

\begin{proof}
It suffices to have such a window length~$L$ that the duration of a Main stage is greater than the largest number of packets that might be injected in a window, which is $\rho L + \beta$. 
Let us assume temporarily that $\rho$ and $\beta$ are known to the stations and therefore the initial value of $L$ can be properly determined, based on these $\rho$ and $\beta$.
This assumption may be dropped, as we show later.
	
A Main stage has at least $L-9n^3\lg L$ rounds, so it suffices for a window size~$L$ to satisfy the following inequality:
\[
L- 9 n^3 \lg L \ge \rho L + \beta
\ . 
\]
The above inequality holds for $L=\frac{9n^3\lg^2 n+\beta}{1-\rho}$, for $n$ that is sufficiently large with respect to $\rho$ and $\beta$, as can be verified directly.
The latency is at most $2L$ because a packet may spend two consecutive windows in a queue, first as a new packet and then as an old one.
This completes the analysis in the case when $L$ is properly set at the beginning of an execution.
	
Next we incorporate into the analysis the mechanism by which the length of the next window may get increased after a current window is over. 	
If the window size is not increased at the end a time window $W$, then all packets injected before~$W$ are delivered during~$W$ and therefore only packets injected during $W$ are present in queues at the end of~$W$. 
Then our estimates of window size from the beginning of the proof apply. 

Suppose the size of a time window~$W$ is greater than the window size of the immediately preceding window. 
In general, let $W_1, W_2, W_3, \ldots, W_k$ be a sequence of consecutive windows, where $W_{i-1}$ occurs directly after $W_i$, for each $i$, and window~$W$ occurs directly after $W_1$.
Let moreover this sequence be  such that the size of the window $W_{i-1}$ is greater than the window size of the window~$W_i$, for each $1<i\le k$. 
This means that the window size was increased at the end of~$W_i$, and also either $W_k$ is the first window of the considered execution  or the size of the window preceding $W_k$ is equal to the size of the window~$W_k$. 
Thus, all packets injected before $W_k$ are delivered during $W_k$. 
Therefore, as the window size of $W_i$ is twice as large as the window size of $W_{i+1}$ for $1\le i<k$, the number of packets in all queues at the beginning of the window $W$ is at most the following:
\[
\rho L\cdot \Bigl(\frac12+\frac1{2^2}+\cdots+\frac1{2^k}\Bigr)+\beta \le \rho L+\beta
\ .
\]
In each execution, the window size $L$ eventually  becomes sufficiently large to provide that all packets injected before a window are transmitted within this window,  and this final window size is at most 
\[
L=\frac{9n^3\lg^2 n+\beta}{1-\rho}
\ , 
\]
by an argument applied as in the first part of this proof.
The latency is again at most twice the length of such a longest window.
\end{proof}

\section{Energy-Oblivious Indirect Routing}

\label{sect:oblivious-indirect}

Let an integer $k<n$ denote an energy cap.
We present now a plain-packet $k$-energy-oblivious algorithm called \algoname{$k$-Cycle}.

The algorithm operates as follows.
Up to $k$ stations are switched on in each round.
The stations are partitioned into $\ell$ \emph{groups} of size $k$ each.
The $i$th group is denoted as $G_i$, for $1,\ldots,\ell$.
Group $G_1$ consists of the $k$ stations $0,1,\ldots, k-1$, the next group $G_2$ comprises the station~$k-1$ and the next $k-1$ stations $k, k+1,\ldots, 2k-2$, the next group $G_3$ includes station~$2k-2$ and the next $k-1$ stations $2k-1,2k,\ldots, 3k-3$, and so on, with the last group padded with dummy stations if needed.
The underlying idea is that a group consists of $k$ stations with consecutive numbers, and a group $G_{i+1}$ starts from the last station in group~$G_i$ and includes the next $k-1$ stations.
In general, the number of groups is at most $\ell\le \frac{n-1}{k-1}+1$.

If $n\le 2k$ then  $k$ gets decreased such that $2k=n+1$, which allows to keep fewer stations switched on.
After this, we may assume that the inequality $2k\le n+1$ holds in general, which implies that there are at least two groups.

Two consecutive groups share one station, called a \emph{connector} of these groups, with group $G_\ell$ sharing  station~$0$ as a connector with group~$G_1$.
The stations in a group are ordered by their names into an ordered cycle.
Groups themselves are also arranged into an ordered cycle as follows: group~$G_{i+1}$ follows group~$G_i$, for $i<\ell$, and group~$G_1$ follows $G_\ell$. 

In each round $t$ of an execution, all the stations in some group $G_i$ are switched on, with the other stations switched off; we say that group $G_i$ is \emph{active} in round~$t$.
The pattern of activity among the groups follows round robin according to the order cycle of the groups.
A group is active for a time segment of these many rounds:
\begin{equation}
\label{eqn:group-time-segment}
\delta=\frac{4 (n-1) k}{n-k}
\ .
\end{equation}
When this time segment ends, the next group in the cyclic order takes over.

Each group executes an algorithm related to the broadcast algorithm \algoname{Old-First-Round-Robin-Withholding (OF-RRW)} during the consecutive rounds the group is active.
Algorithm \algoname{OF-RRW} was considered in~\cite{AnantharamuCKR-JCSS19}, we adapt it as a building block of routing algorithms; the details of the adaptation are given next.

There is a conceptual token associated with each group. 
The actions of stations  in a group are controlled by feedback from the channel.
The feedback is the same for all the stations in a group, which allows to handle the token in such a manner that it is not duplicated nor  lost.
The token passes through all the stations in a group in a round-robin manner.
When the token completes the whole cycle then this also ends a \emph{phase}.
Packets injected or adopted during a phase are \emph{new} for this phase and otherwise they are \emph{old} for this phase.
When a station receives a token then it transmits all its old packets one by one.
If there is no old packet to transmit by a station holding the token then this  station does not transmit anything, which results in a silent round.
A silent round triggers the token to advance to the next station in the group in their cyclic order.
When a station holding the token of a group~$G_i$ transmits then the message is heard on the channel by all the stations in the group~$G_i$.
If the destination station of this packet belongs to $G_i$ then the packet gets delivered and otherwise the station in $G_i$ that is a connector with $G_{i+1}$ adopts the packet and becomes its  relay.
This mechanism of handling packet implies that a packet may hop through all the $\ell$ groups until it reaches its destination station.

%: theorem

\begin{theorem}
\label{thm:Cycle-of-Groups}

Algorithm \algoname{$k$-Cycle} routes packets correctly, when the energy cap is at least $k$, and has latency at most $(32+\beta)\cdot n$ against an $(\rho,\beta)$-adversary such that $\rho< \frac{k-1}{n-1}$.
\end{theorem}

\begin{proof}
A group operates as a virtual cycle of $k$ stations executing broadcast algorithm \algoname{OF-RRW}.
Such a cycle in isolation would have broadcast latency at most 
\begin{equation}
\label{eqn:cycle-1}
\frac{2}{1-\rho}\cdot k+\beta(1+\rho)\le \frac{2k}{1-\rho}+2\beta
\ , 
\end{equation}
for an injection rate satisfying only the inequality~$\rho<1$; see~\cite{AnantharamuCKR-JCSS19}.
A packet may perform at most $\frac{n-1}{k-1}$ hops through consecutive groups, which effectively amplifies injection rate by this factor.
Therefore injection rates need to be less than $\frac{k-1}{n-1}$ to make routing stable.

The bound~\eqref{eqn:cycle-1} on packet delay has two parts, among which $\frac{2k}{1-\rho}$ applies to packets injected within the injection-rate component and $2\beta$ applies to the packets for which injection-rate component does not suffice and they need the adversary's burstiness to justify their injection; see~\cite{AnantharamuCKR-JCSS19}.
We consider the delay of packets by categorizing the packets into two groups:  those for which the injection-rate component $\frac{2k}{1-\rho}$ suffices and the remaining ones to which the burstiness component $2\beta$ needs to apply to justify their delay.
This categorization of packets if for accounting purposes only.

For packets accounted for as injected subject to the injection-rate constraint,  the bound $\frac{2k}{1-\rho}$ on packet delay becomes at most
\begin{equation}
\label{eqn:cycle-2}
 \frac{2 k(n-1)}{n-k}
\end{equation}
after combining it with the upper bound $\rho<\frac{k-1}{n-1}$ on injection rates.
Bound~\eqref{eqn:cycle-2} on packet delay is less than the duration of a continuous segment of rounds of activity of a group of stations determined by~\eqref{eqn:group-time-segment}.
This implies that all the packets held by the stations in a group, and accounted for as injected subject to the injection-rate constraint, are heard on the channel when the group becomes active. 
This also means that these among such packets that are addressed to other groups will hop through the connector to the next group, while their current group is active.
Such hops will continue without delay other than that incurred by the period of activity of the group where these packets reside. 

The bound~\eqref{eqn:cycle-2} needs to be increased to the duration~$\delta$ for a period of activity~\eqref{eqn:group-time-segment} of a group and then multiplied by the number of hops a packet can make, to obtain a bound on latency of routing. 
This yields the following estimate
\[
 \frac{4 k(n-1)}{n-k}\cdot \frac{n-1}{k-1}\le \frac{8(n-1)^2}{n-k}\le 16 (n-1)
 \ ,
\]
assuming $2k\le n+1$.
This bound accounts for a full cycle of activity of all the groups, but a packet may spend another such cycle waiting for the group into which is got injected to become active.
This means that $32\cdot n$ is a bound on latency, restricted to packets that can be accounted for as injected subject to the injection-rate restriction.

Next, we estimate the delay of packets that need the adversary's burstiness to account for their injection.
A half of the duration $\delta=\frac{4 (n-1) k}{n-k}$ of group's activity  is needed for packets that can be accounted for as injected subject to the injection-rate restriction, as estimated by~\eqref{eqn:cycle-2}.
What remains are $\frac{2 (n-1) k}{n-k}$ rounds that can be used to transmit a surplus of packets due to a burst of injections.
Out of these many rounds, at most $k$ can be waisted because the token visits stations without packets.
What remains are at least these many rounds:
\begin{equation}
\label{eqn:burstiness}
\frac{2 (n-1) k}{n-k} - k
\ .
\end{equation}
Observe that if $2k\le n+1$ then $\frac{2 (n-1) }{n-k}\ge 4$, and so the quantity~\eqref{eqn:burstiness} is at least~$3k$.
The packet delay of the considered packets is thus at most $\frac{2\beta}{3k}$ multiplied by the number of groups, which is at most the following:
\[
\frac{2\beta}{3k}\cdot \Bigl(\frac{n-1}{k-1}+1\Bigr)
\le 
\frac{2\beta}{3k}\cdot\frac{n+k-2}{k-1}
\le
\frac{2\beta}{3k}\cdot\frac{\frac{3}{2}n-2}{k-1}
\le
\beta n
\ .
\]
The bound on latency is the maximum of the partial upper bounds $32n$ and $\beta n$.  
\end{proof}

Next we give an impossibility which demonstrates that the bound on injection rate in Theorem~\ref{thm:Cycle-of-Groups} is very close to optimal.

%: theorem 

\begin{theorem}
\label{thm:indirect-oblivious}

For each $n$ and $k<n$, a $k$-energy-oblivious routing algorithm is unstable against adversaries with injection rates greater than $\frac{k}{n}$.
\end{theorem}

\begin{proof}
If a station is switched on in a round then this contributes one \emph{station-round}.
A contiguous time interval $\tau$ of $|\tau|=t$ rounds can contribute at most $kt$ station-rounds. 
By the double-counting principle,  there is some station~$v$ which is switched on for at most $\frac{kt}{n}$ rounds during these $t$ rounds.
The adversary with injection rate~$\rho$ can inject at least $\rho t$ packets into station~$v$ during these rounds.
Even if $v$ transmits successfully in each round in $\tau$ then it can transmit at most $\frac{kt}{n}$ packets.
If $\rho>\frac{k}{n}$ then there remain at least these many packets that need to be queued by~$v$:
\[
\rho t-\frac{kt}{n} = t\Bigl(\rho-\frac{k}{n}\Bigr)
\ .
\]
This number can be made arbitrarily large for a suitably large~$t$.
\end{proof}

\begin{corollary}
There is no $k$-energy-oblivious universal algorithm  when $k=c\cdot n$, for a constant $c<1$.
\end{corollary}

\begin{proof}
By Theorem~\ref{thm:indirect-oblivious}, a  $k$-energy-oblivious algorithm is unstable for injection rates that are greater than the ratio $\frac{k}{n}=c<1$.
\end{proof}

\section{Energy-Oblivious Direct Routing}

\label{sect:oblivious-direct}

Let an integer $k<n$ denote the energy cap.
Now  we present a plain-packet $k$-energy-oblivious algorithm that routes packets directly.
It is called \algoname{$k$-Clique}.
There are up to $k$ stations switched on in each round.
We assume that $k$ is even and divides~$2n$, to simplify the notation.
The stations are partitioned into $\frac{2n}{k}$ disjoint sets of size~$\frac{k}{2}$ each.
These sets are combined in $\frac{n}{k}(\frac{2n}{k}-1)$ \emph{pairs} of size~$k$ each.
There are at least $3$ pairs, assuming $\frac{k}{2}\le \frac{n}{3}$.
If $\frac{k}{2}> \frac{n}{3}$ then we can decrease~$k$ by keeping fewer stations switched on, so that the inequality $k\le \frac{2n}{3}$  holds.

In each round $t$ of an execution, all the stations in some pair are switched on, with the other stations switched off; we say that the pair is \emph{active} in round~$t$.
The pairs are arranged into a virtual cycle to assign them the rounds of activity in a round-robin manner.
A pair is active for one round at a time, and then the next pair takes over.
When a pair is active then its stations execute an algorithm based on the principle of broadcasting algorithm \algoname{OF-RRW}.
A station that has the token transmits all the old packets whose destinations are  among the stations that make up the pair.

%: theorem

\begin{theorem}
\label{them:k-clique}

If algorithm \algoname{$k$-Clique} is executed against a $(\rho,\beta)$-adversary then it has a bounded latency for injection rates $\rho<\frac{k^2}{n(2n-k)}$, and the latency is at most $8 \frac{n^2}{k}(1+\frac{\beta}{2k})$ if the injection rate~$\rho$ is at most $\frac{k^2}{2n(2n-k)}$.
\end{theorem}

\begin{proof}
Let $m=\frac{n(2n-k)}{k^2}$ be the number of pairs.
We will use the  estimate $m\le 2\frac{n^2}{k^2}$.
A strategy for the adversary that maximizes queues and latency works by  injecting packets into one pair with destinations in the same pair as well.
Since a pair is allotted one round out of a segment of rounds equal to the number of pairs, an injection rate needs to be less than the inverse of the number of pairs, which is $\frac{1}{m}=\frac{k^2}{n(2n-k)}$.
For each pair, when time is scaled only to the rounds which are assigned for the pair to execute \algoname{OF-RRW}, the injection rate is less than~$1$, so the algorithm has bounded latency.

Suppose that the inequality $\rho<\frac{1}{m}$ holds.
A pair of $k$ stations operating in isolation, and with time scaled only to the rounds assigned to the pair to be active, would have effective injection rate~$m\cdot  \rho$ and so its latency would be at most the following
\begin{equation}
\label{eqn:clique}
\frac{2}{1-m\rho}\cdot k+\beta(1+m\rho)\le \frac{2}{1-m\rho}\cdot k+2\beta 
\ ,
\end{equation}
by applying the bound on broadcast latency of \algoname{OF-RRW} derived in~\cite{AnantharamuCKR-JCSS19}.
The bound needs to be increased by a multiplicative factor of~$m$, since a pair  operates in one round only in a segment of $m$ rounds.
This gives the following estimate on latency:
\[
\frac{2m}{1-m\rho}\cdot k+2\beta m
\le
\frac{2n^2}{k}\cdot \frac{2}{1-m\rho}+ \frac{4\beta n^2}{k^2}
\ ,
\]
which holds for any injection rate satisfying $\rho<m^{-1}$.
Assuming additionally that the inequality 
\[
\rho n(2n-k) \le \frac{k^2}{2}
\]
holds, we can use the following estimate:
\[
\frac{2}{1-m\rho} 
=
\frac{2}{1-\rho \frac{n(2n-k)}{k^2}}
=
\frac{2 k^2}{k^2-\rho n(2n-k)}
\le 4
\ .
\]
This yields $8 \frac{n^2}{k}(1+\frac{\beta}{2k})$ as a bound on latency.
\end{proof}

\Paragraph{Maximum throughput of energy-oblivious direct routing.}

We describe a direct-routing $k$-energy-oblivious algorithm achieving the throughput  $\frac{k(k-1)}{n(n-1)}$ for any burstiness~$\beta$.
The algorithm is called \algoname{$k$-Subsets}.
It uses algorithm \algoname{Move-Big-To-Front (MBTF)}~\cite{ChlebusKR-DC09} as a subroutine.
\algoname{MBTF} provides stability for injection rate~$1$ with any burstiness,  for a multiple access channel \emph{without} any energy cap.

Let us fix an enumeration of all the $k$-element subsets of the set $[n]$ in order: $A_0,\ldots,A_{\gamma-1}$, where $\gamma=\binom{n}{k}$. 
For an integer~$i$ such that $0\le i \le \gamma-1$ and an integer $j\ge 0$, the rounds of the form $i+j\gamma$ make \emph{thread~$i$}. 
Algorithm \algoname{MBTF} operates in $\gamma$ instantiations corresponding to the threads.
Each such an instantiation has a dedicated queue in every station.
The stations in $A_i$ are active during thread~$i$ and process packets assigned to this thread.
An execution is structured into \emph{phases}, each of length $\gamma$, such that each thread has one round in a phase.
A packet injected during a phase~$j$ is treated by an instantiation of \algoname{MBTF} that handles it as if it were injected ``at round~$j$.''
These specifications mean that the algorithm is $k$-energy-oblivious.

At the beginning of a phase, a station~$v$ assigns all the packets it received in the previous phases to the threads in the following manner.
For each station~$w$, station~$v$ keeps track of the numbers $x_0(w),\ldots,x_{\gamma-1}(w)$, which represent the respective numbers of packets addressed to~$w$ and already allocated by~$v$ to threads $0,\ldots,\gamma-1$. 
Station~$v$ allocates the packets addressed to~$w$ that it received in the previous phase such  that the resulting allocation is as balanced as possible, subject to the constraint that a packet addressed to~$w$ can be allocated to a thread~$i$ only if both stations $v$ and~$w$ are in the set~$A_i$. 
The algorithm routes directly, since when a packet is heard transmitted in a round assigned to a thread~$i$, the receiver is switched on by virtue of belonging to~$A_i$.
Obtaining balancing allocations means that we want the numbers $x_0(w),\ldots,x_{\gamma-1}(w)$ differ among themselves as little as possible; it follows that these numbers differ by at most~$1$ at the beginning of each phase.

%: theorem

\begin{theorem}
\label{thm:oblivious-direct-throughput}
	
For each  $k<n$, algorithm \algoname{$k$-Subsets} is stable against adversaries with injection rate $\frac{k(k-1)}{n(n-1)}$ and the number of queued packets is at most $2\binom{n}{k} (n^2+\beta)$ in every round.
\end{theorem}

\begin{proof}
Let $\lambda=\frac{k(k-1}{n(n-1)}$ denote the injection rate we consider.
Suppose that the algorithm is not stable for injection rate~$\lambda$, to arrive at a contradiction.
There exists a thread $i$ in which some queue corresponding to packets arriving at station~$v$ with address $w$ and assigned to this thread grows unbounded. 
Since algorithm \algoname{MBTF} is stable for injection~$1$ and a fixed burstiness, there exists an infinite sequence of rounds $t_1,t_2,\ldots,t_j,\ldots$, for all $j\ge 1$, such that the number of packets from station~$v$ to~$w$ that get assigned to thread~$i$ by round~$j$ is at least $|t_j|/\gamma + j+2$.
Indeed, each thread is executed once every~$\gamma$ rounds, so the execution of algorithm  \algoname{MBTF} in thread~$i$ would be stable if  burstiness were bounded.  

The algorithm allocating packets to threads guarantees that the number of packets with the same pair $(v,w)$, of the source~$v$ and destination~$w$, assigned by $v$ to threads with stations $v$ and~$w$ being active is almost balanced in each time period, in that the difference between any two of them is either $-1$ or $0$ or~$1$. 
A thread with this property is determined by the stations different from $v$ and $w$, so their number equals the following:
\[
\binom{n-2}{k-2} = \binom{n}{k}\cdot \frac{k(k-1}{n(n-1)}=\gamma\lambda
\ .
\] 
The number of packets injected to $v$ addressed to $w$ by round~ $t_j$, for every $j\ge 1$, is at least $(t_j/\gamma + j+2)-1$ multiplied by the number of threads handling packets from~$v$ to~$w$. 
It follows that the following  is a lower bound on the number of packets addressed to~$w$ that are assigned by~$v$ to some threads by round~$t_j$:
\[
\lambda t_j + \gamma\lambda (j+1)=\lambda(t_j+\gamma) + \gamma\lambda j
\ .
\]
Therefore, the number of packets that arrive at~$v$ by round $t_j+\gamma$  is at least $\lambda(t_j+\gamma) + \gamma\lambda j$. 
This contradicts the restrictions on the adversary for sufficiently large~$j$, as the burstiness would be exceeded by round~$t_j+\gamma$, and completes showing stability.
A bound on the number of queued packets follows from the respective bound for algorithm \algoname{MBTF} given in~\cite{ChlebusKR-DC09}, which is applied independently for each thread.
\end{proof}

It may occur  in an execution of algorithm \algoname{$k$-Subsets} that some packets never get delivered and so remain queued forever. 
The algorithm can be modified to prevent this as long as injection rates are less than $\frac{k(k-1}{n(n-1)}$.
Namely, it is sufficient to replace algorithm \algoname{MBTF}, used as a procedure in \algoname{$k$-Subsets}, by \algoname{Round-Robin-Withholding (RRW)}, see \cite{ChlebusKR-TALG12}.
The resulting algorithm is stable for any injection rate less than~$\frac{k(k-1)}{n(n-1)}$ and achieves bounded latency.
By the performance bounds of \algoname{RRW}, see \cite{AnantharamuCKR-JCSS19}, the latency is $\Theta(\gamma \cdot (n+\beta))$, for a fixed adversary with injection rate less than~$\frac{k(k-1)}{n(n-1)}$.
The latency bound is at least $\gamma=\binom{n}{k}$, which is exponential in~$n$ when $k$ is linear in~$n$.

\Paragraph{A lower bound on throughput.}

We give a matching lower bound on throughput, which demonstrates that the throughput in Theorem~\ref{thm:oblivious-direct-throughput} is maximum achievable in the class of energy-oblivious algorithms.

%: theorem

\begin{theorem}
For each integer $n$ and $k<n$, and for any $k$-energy-oblivious  algorithm routing directly, and for every adversary with an injection rate greater than~$\frac{k(k-1)}{n(n-1)}$, some executions of the algorithm may be unstable against this adversary.
\end{theorem}

\begin{proof}
We will count the following quantities: for each ordered pair $(x,y)$ of different stations $x$ and~$y$, if they  are switched on simultaneously in a round then this rounds contributes one \emph{station-pair round}.
A contiguous time interval $\tau$ of $|\tau|=t$ rounds can contribute at most $k(k-1)t$ station-pair rounds. 
By the double-counting principle,  there is some ordered pair of stations~$(w,z)$ such that $w$ and $z$ are switched on together for at most $\frac{k(k-1)}{n(n-1)}\cdot t$ rounds in the time interval~$\tau$.
The adversary with injection rate~$\rho$ can inject at least $\rho t$ packets into station~$w$ during these rounds.
Let all these packets be destined for~$z$.
Even if $w$ transmits successfully in each round in $\tau$ such that $w$ is switched on along with $z$, then it can transmit at most $\frac{k(k-1)}{n(n-1)}\cdot t$ packets.
If $\rho>\frac{k(k-1)}{n(n-1)}$ then there remain at least these many packets that need to be queued by~$w$:
\[
\rho t-\frac{k(k-1)}{n(n-1)}t = t(\rho-\frac{k(k-1)}{n(n-1)})
\ .
\]
This number can be made arbitrarily large for a suitably large~$t$.
\end{proof}

\section{Conclusion}

There are several natural questions related to the presented topic that are open.
One of them is to derive tradeoffs between latency and energy cap,  for a class of algorithms.
Such tradeoffs are natural to hold, since small energy cap is a constrain on scheduling transmissions in a distributed manner.

Another group of questions pertains to minimizing latency so that it is $\cO(n)$.
Broadcast algorithms with such latency for all injection rates less than~$1$ have been developed, in the case when there are no energy constraints on a channel.
It is not known if there exists a universal routing algorithm for a non-trivial bound on energy cap that has $\cO(n)$ latency for each injection rate less than~$1$.
It is not known if a plain-packet routing algorithm for a constant energy cap  can have latency $\cO(n)$ for a non-trivial range of injection rates.
It is not known if an algorithm  routing directly and subject to a non-trivial energy cap can have latency $\cO(n)$.

\bibliographystyle{plain}

\bibliography{energy-mac-route}

\end{document}